\documentclass{article}

\usepackage{amsmath,amsfonts,amssymb,amsthm,mathtools}
\usepackage{amscd}
\usepackage{bbm}
\usepackage{enumerate}
\usepackage{galois}
\usepackage{mathrsfs}
\usepackage{xypic}
\usepackage{geometry}
\usepackage{hyperref}
\usepackage{booktabs}
\usepackage{simplewick}

\usepackage{color}

\newcommand{\lm}{\lambda}
\newcommand{\lmd}{\lambda}  
\newcommand{\clm}{\Lambda}
\newcommand{\Lmd}{\Lambda}   
\newcommand{\p}{\partial}
\newcommand{\al}{\alpha}      
\newcommand{\afa}{\alpha}
\newcommand{\tM}{\tilde{M}}
\newcommand{\tL}{\tilde{L}}
\newcommand{\ve}{\varepsilon}

\newcommand{\mcalJ}{\mathcal{J}}
\newcommand{\mcalP}{\mathcal{P}}
\newcommand{\mcalQ}{\mathcal{Q}}

\newcommand{\vsp}{\vspace{0.8ex}}
\newcommand{\vvsp}{\vspace{0.5ex}}

\newcommand{\beq}{\begin{equation}}
\newcommand{\eeq}{\end{equation}}

\newcommand{\dtafrac}[2]{\frac{\delta#1}{\delta #2}}

\DeclareMathOperator{\res}{Res}
\DeclareMathOperator{\nd}{d\!}
\DeclareMathOperator{\td}{d\!}

\newtheorem{thm}{Theorem}[section]
\newtheorem{rmk}[thm]{Remark}

\newtheorem{lem}[thm]{Lemma}
\newtheorem{prop}[thm]{Proposition}
\newtheorem{defn}{Definition}[section]
\newtheorem{ex}{Example}[section]

\numberwithin{equation}{section}

\begin{document}

\title{ Tri-Hamiltonian Structure of the Ablowitz-Ladik Hierarchy}

\author{
{Shuangxing Li, Si-Qi Liu, Haonan Qu, Youjin Zhang}\\
{\small Department of Mathematical Sciences, Tsinghua University}\\
{\small Beijing 100084, P. R. China}
\date{}
}
\maketitle

\begin{abstract}
We construct a local tri-Hamiltonian structure of the Ablowitz-Ladik hierarchy, and compute the central invariants of the associated bihamiltonian structures. We show that the central invariants of one of the bihamiltonian structures are equal to $\frac1{24}$, and the dispersionless limit of this bihamiltonian structure coincides with the one that is defined on the jet space of the Frobenius manifold associated with the Gromov-Witten invariants of local $\mathbb{CP}^1$. This result provides support for the validity of Brini's conjecture on the relation of these Gromov-Witten invariants with the Ablowitz-Ladik hierarchy.
\vskip 0.3true cm
\noindent {\bf{Keywords}}\ \ Ablowitz-Ladik hierarchy; tri-Hamiltonian structure; bihamiltonian structure; central invariant; Gromov-Witten invariant
\end{abstract}

\tableofcontents
\section{Introduction}
The Ablowitz-Ladik equation is one of the most important differential-difference nonlinear integrable equations in soliton theory. It was obtained by Ablowitz and Ladik as a discretization of the nonlinear Schr\"odinger equation and was solved by the inverse scattering method in \cite{AL-95, AL-96}, and has close relations to some other important soliton equations such as the 2D Toda lattice equation \cite{Vekslerchik}, the relativistic Toda lattice equation \cite{Suris}, the Toeplitz lattice equation that appears in the study of random matrix models \cite{Adler-Moerbeke}. It also characterizes some elementary geometric properties of the motion of discrete curves on the $N$-dimensional spheres \cite{Doliwa}. The importance of the Ablowitz-Ladik equation is manifested in the recent study of the Gromov-Witten invariants of local $\mathbb{CP}^1$ \cite{Brini-1, Brini-2, Brini-3}. As it is conjectured by Brini in \cite{Brini-1} that the generating function of these Gromov-Witten invariants is given by a particular tau-function of the Ablowitz-Ladik hierarchy. The purpose of the present paper is to study the multi-Hamiltonian structure of the Abolowitz-Ladik hierarchy, and to answer the question raised by Brini and his collaborators in \cite{Brini-2, Brini-3} on the existence of a bihamiltonian structure of the Ablowitz-Ladik hierarchy.

The Ablowitz-Ladik equation is given by \cite{AL-96}
\[i U_{n,t}=U_{n+1}+U_{n-1}-2 U_n\pm U^*_n U_n (U_{n+1}+U_{n-1}),\quad n\in\mathbb{Z},\]
it is the compatibility condition of the discrete Zakharov-Shabat or AKNS spectral problem \cite{AKNS, AL-95, AL-96, Zakharov-Shabat}
\begin{equation}\label{x-part}
\Lambda\begin{pmatrix}\psi_1\\ \psi_2\end{pmatrix}=\begin{pmatrix}\lambda & q \\ r & \lambda^{-1}\end{pmatrix}\begin{pmatrix}\psi_1\\ \psi_2\end{pmatrix}
\end{equation}
and the evolution of the eigenfunctions
\begin{equation}\label{t-part}
\begin{pmatrix}\psi_1 \\ \psi_2 \end{pmatrix}_t=\begin{pmatrix}F & G \\ H & K\end{pmatrix}
\begin{pmatrix}\psi_1\\ \psi_2\end{pmatrix},
\end{equation}
where $\Lambda$ is the shift operator defined by $\Lambda(f)(n)=f(n+1)$, $\psi_1, \psi_2, q, r$ are functions of $n, t$ with
\[r(n,t)=\mp q^*(n,t)=U_n(t),\]
and
\begin{align*}
&F=i \lm^2-i q r^--i,\quad G=i q \lm-i q^- \lm^{-1},\\
&H=i r^- \lm-i r \lm^{-1},\quad K=i+i q^- r-i\lm^{-2}.
\end{align*}
Here and in what follows we use the notation
\[\Lambda(f)(n)=f(n+1)=f^+,\quad \Lambda^{-1}(f)(n)=f(n-1)=f^-.\]

For general independent unknown functions $q(n), r(n)$, we can
assume that the evolution \eqref{t-part} of the eigenfunctions is given the functions $F, G, H, K$ which are some appropriate polynomials in $\lm, \lm^{-1}$.
The coefficients of these polynomials can be obtained by the compatibility condition of \eqref{x-part} and \eqref{t-part}. This compatibility condition also yields a system of differential-difference equations for $q(n,t), r(n,t)$, and the set of differential-difference equations thus obtained is called the Ablowitz-Ladik hierarchy in \cite{Suris}.
The following three choices of the
evolutions \eqref{t-part} of the eigenfunctions are given in \cite{Suris} as elementary flows of the Ablowitz-Ladik hierarchy:
\begin{align}\label{elementary-f}
&\begin{pmatrix}\psi_1 \\ \psi_2 \end{pmatrix}_{t_{-1}}=\begin{pmatrix}0 & -q^{-}\lm^{-1} \\ -r \lm^{-1} & -\lm^{-2}+q^- r\end{pmatrix}
\begin{pmatrix}\psi_1\\ \psi_2\end{pmatrix},\quad
\begin{pmatrix}\psi_1 \\ \psi_2 \end{pmatrix}_{t_{0}}=\begin{pmatrix}-1 & 0 \\ 0 & 1\end{pmatrix}
\begin{pmatrix}\psi_1\\ \psi_2\end{pmatrix},\\
&\begin{pmatrix}\psi_1 \\ \psi_2 \end{pmatrix}_{t_{1}}=\begin{pmatrix}\lm^2-q r^- & q \lm \\ r^- \lm & 0\end{pmatrix}
\begin{pmatrix}\psi_1\\ \psi_2\end{pmatrix}.\label{elementary-pf}
\end{align}
The corresponding differential-difference equations are given by \cite{Suris}
\begin{equation}
\begin{cases} q_{t_{-1}}=q^- (1-q r)\\ r_{t_{-1}}=-r^+(1-q r)\end{cases},\quad
\begin{cases} q_{t_{0}}=-2 q\\ r_{t_0}=2 r\end{cases},
\quad
\begin{cases} q_{t_1}=q^+ (1-q r)\\ r_{t_1}=-r^-(1-q r)\end{cases}.
\end{equation}
Note that the flow $\frac{\p}{\p t}$ given by \eqref{t-part} can be represented as
\[\frac{\p}{\p t}=i \left(\frac{\p}{\p t_{-1}}+\frac{\p}{\p t_0}+\frac{\p}{\p t_{1}}\right).\]
We can obtain more elementary flows by replacing the r.h.s. of the
$t_{-1}$-evolution of \eqref{elementary-f} by higher order polynomials in $\lm^{-1}$,
or by replacing the r.h.s. of the $t_1$-evolution of \eqref{elementary-pf} by higher order polynomials in $\lm$, and we call them the negative and positive flows of the Ablowitz hierarchy respectively. These flows can in fact be generated by a recursion operator acting on the $t_0$-flow.

The flows of the Ablowitz-Ladik hierarchy can be represented as the Hamiltonian systems with Hamiltonian operator
\[\mathcal{P}=\begin{pmatrix} 0 & qr -1\\ 1-qr &0\end{pmatrix}.\]
The Hamiltonians of the flows $\frac{\p}{\p t_{-1}}$, $\frac{\p}{\p t_0}$, $\frac{\p}{\p t_1}$ are given by
\[H_{-1}=\sum_{k\in\mathbb{Z}} q(k) r(k+1),\quad H_0=\sum_{k\in\mathbb{Z}}\log(1-q(k) r(k)),\quad
H_1=\sum_{k\in\mathbb{Z}} q(k+1) r(k).\]

A bihamiltonian structure is considered in \cite{Ercolani}  for the Ablowitz-Ladik hierarchy with the first Hamiltonian structure given above.
However, the second Hamiltonian structure given there is nonlocal. The main purpose of the present paper is to study the local multi-Hamiltonian structure of Ablowitz-Ladik hierarchy by using a pair of new unknown functions
\begin{equation}\label{pq-def}
P=\frac{q}{q^-}, \quad Q=\frac{q}{q^-}(1-q^-\,r^-),
\end{equation}
and to show that in terms of these unknown functions this integrable hierarchy possesses a local tri-Hamiltonian structure $(\mathcal{P}_1, \mathcal{P}_2, \mathcal{P}_3)$. An important property of this tri-Hamiltonian structure is that the bihamiltonian structure $(\mathcal{P}_1, \mathcal{P}_2)$ contained in it possesses a  dispersionless limit which coincides with the bihamiltonian structure of hydrodynamic type that is defined on the jet space of a generalized Frobenius manifold \cite{Brini-2}. This Frobenius manifold is almost dual to the one that is associated to the Gromov-Witten invariants of local $\mathbb{CP}^1$ \cite{Brini-1, Brini-2, Brini-3}. We show that all the central invariants \cite{DLZ-2006, LZ-2005} of the bihamiltonian structure $(\mathcal{P}_1, \mathcal{P}_2)$ equal to $\frac1{24}$. Together with the results of classification of infinitesimal deformations of semisimple bihamiltonian structures of hydrodynamic type given in \cite{DLZ-2006, LZ-2005}, this fact on the central invariants of the bihamiltonian structure $(\mathcal{P}_1, \mathcal{P}_2)$ suggests that the Ablowitz-Ladik hierarchy should be the topological deformation of its dispersionless limit, and it provides a strong support for the validity of Brini's conjecture on the close relation between the Ablowitz-Ladik hierarchy and the Gromov-Witten invariants of local $\mathbb{CP}^1$ \cite{DLZ-2006, DLZ-2018, LZ-2005}. We also compute the central invariants of other bihamiltonian structures $(\mathcal{P}_i, \mathcal{P}_j)$ that are contained in the tri-Hamiltonian structure, and find that the central invariants of the bihamiltonian structure $(\mathcal{P}_3, \mathcal{P}_2)$ are equal to $-\frac1{24}$. This fact together with properties of the leading terms of $(\mathcal{P}_3, \mathcal{P}_2)$ implies that this bihamiltonian structure is equivalent to $(\mathcal{P}_1, \mathcal{P}_2)$,  and it also leads us to the construction of a B\"acklund transformation of the Ablowitz-Ladik equation. Part of the results of this paper is presented in \cite{Li}.

We organize the paper as follows. In Sec.\,2 we derive the positive and negative flows of the Ablowitz-Ladik hierarchy represented in terms of the unkown functions $P, Q$. In Sec.\,3 and Sec.\, 4, we show that the Ablowitz-Ladik hierarchy has a local tri-Hamiltonian structure. In Sec.\,5, we compute the central invariants of the bihamiltonian structures that are given by the tri-Hamiltonian structure. In Sec.\,6 and Sec.\,7 we study the relation between the positive and negative flows of the Ablowitz-Ladik hierarchy, and establish an equivalence relation between
the bihamiltonian structures $(\mathcal{P}_3, \mathcal{P}_2)$ and $(\mathcal{P}_1, \mathcal{P}_2)$. Sec.\,8 is a conclusion.

\section{The Ablowitz-Ladik hierarchy}
Let us first note that the spectral problem \eqref{x-part} can be equivalently represented as
\begin{equation}\label{sp-1}
\psi_1^+-\left(\lambda+\frac{P}{\lambda}\right)\psi_1+Q \psi_1^{-}=0,
\end{equation}
here the functions $P, Q$ are defined by \eqref{pq-def}.
Denote $\phi=\lm^{n-1} \psi_1$, then we can rewrite \eqref{sp-1} in the form
\beq\label{sp-2}
L\phi=z\phi,\quad z=\lm^2,
\eeq
where
\begin{align}
L&=(1-Q\clm^{-1})^{-1}(\clm-P)\notag\\
&=\clm+Q-P+Q(\clm-Q)^{-1}(Q-P)\notag\\
&=\clm+Q-P+Q(Q^--P^-)\clm^{-1}+QQ^{-}(Q^{--}-P^{--})\clm^{-2}+\dots.\label{op-L}
\end{align}
\begin{defn}
We call the hierarchy of differential-difference equations given by the
Lax equations
\begin{align}
\frac{\p L}{\p t_k}&=\frac1{(k+1)!}[(L^{k+1})_+, L],\quad k\ge 0,\label{al-def}\\
\frac{\p L}{\p s_k}&=\frac1{(k+1)!}[(M^{k+1})_-, L],\quad k\ge 0,\label{al-def-n}
\end{align}
the Ablowitz-Ladik hierarchy, where the operators $L$ and $M$
are defined respectively by \eqref{op-L} and
\begin{align}
M&=(\clm-P)^{-1}(1-Q\clm^{-1})\notag\\
&=-\frac1{P}\left(1-\clm \frac1{P}\right)^{-1}(1-Q\clm^{-1})\\
&=\frac{Q}{P}\clm^{-1}+\frac{Q^+}{P P^+}-\frac1{P}+\left(\frac{Q^{++}}{PP^+P^{++}}-\frac1{PP^+}\right)\clm+\dots.\label{op-tL}
\end{align}
We also call the flows $\frac{\p}{\p t_k}$ and $\frac{\p}{\p s_k}$
the positive and negative flows of the Ablowitz-Ladik hierarchy
respectively.
\end{defn}
It was shown in \cite{Brini-2} that the Lax equations \eqref{al-def}, \eqref{al-def-n}
are well-defined and they yield a hierarchy of differential-difference equations
for the unknown functions $P, Q$. Indeed, if we denote
\beq\label{def-AB}
A=\clm-P,\quad B=1-Q\clm^{-1},
\eeq
then
\[L=B^{-1}A,\quad M=A^{-1}B,\]
and the Lax equations \eqref{al-def}, \eqref{al-def-n} are equivalent to the following equations \cite{Brini-2}
\begin{align}
\frac{\p A}{\p t_k}&=\frac1{(k+1)!}\left((\tL^{k+1})_+A-A (L^{k+1})_+\right),\label{eq-A}\\
\frac{\p B}{\p t_k}&=\frac1{(k+1)!}\left((\tL^{k+1})_+B-B (L^{k+1})_+\right),\label{eq-B}\\
\frac{\p A}{\p s_k}&=\frac1{(k+1)!}\left((\tilde{M}^{k+1})_-A-A (M^{k+1})_-\right),\label{eq-C}\\
\frac{\p B}{\p s_k}&=\frac1{(k+1)!}\left((\tilde{M}^{k+1})_-B-B (M^{k+1})_-\right),\label{eq-D}
\end{align}
where
\[\tL=A B^{-1},\quad \tilde{M}=B A^{-1}.\]
By using the identities
\begin{align*}
(\tL^{k+1})_+A-A (L^{k+1})_+
&=-(\tL^{k+1})_-A+A (L^{k+1})_-,\\
(\tL^{k+1})_+B-B (L^{k+1})_+
&=-(\tL^{k+1})_-B+B (L^{k+1})_-,\\
(\tilde{M}^{k+1})_-A-A (M^{k+1})_-
&=-(\tilde{M}^{k+1})_+A+A (M^{k+1})_+,\\
(\tilde{M}^{k+1})_-B-B (\tL^{k+1})_-
&=-(\tilde{M}^{k+1})_+B+B (M^{k+1})_+,
\end{align*}
it is easy to see that equations \eqref{eq-A}-\eqref{eq-D} yield   differential-difference equations for the functions $P, Q$.
If we denote the operators $L^{k}, \tL^{k}$ by
\begin{align*}
L^{k}&=\sum_{l\ge 0} a^{k}_l \clm^{k-l},\quad
\tL^{k}=\sum_{l\ge 0} b^{k}_l \clm^{k-l},\\
M^{k}&=\sum_{l\ge 0} c^{k}_l \clm^{-k+l},\quad
\tilde{M}^{k}=\sum_{l\ge 0} d^{k}_l \clm^{-k+l},
\end{align*}
then the flows of the Ablowitz-Ladik hierarchy can be represented by
\begin{align*}
P_{t_k}&=\frac1{(k+1)!}\left(b^{k+1}_{k+1}-a^{k+1}_{k+1}\right)P,\quad
Q_{t_k}=\frac1{(k+1)!}\left(b^{k+1}_{k+1}-\clm^{-1}(a^{k+1}_{k+1})\right)Q,\\
P_{s_k}&=\frac1{(k+1)!}\left(\clm(c^{k+1}_{k})-d^{k+1}_{k}\right),\quad
Q_{s_k}=\frac1{(k+1)!}\left(c^{k+1}_k-d^{k+1}_k)\right).
\end{align*}

By using the identities
\begin{align*}
&A L^k=\tL^k A=\tL^{k+1}B=B L^{k+1},\\
&B M^k=\tM^k B=\tM^{k+1} A=A M^{k+1},
\end{align*}
we obtain the following relations among the coefficients $a^k_l, b^k_l, c^k_l, d^k_l$:
\begin{align}
&\clm(a^k_{l+1})-P a^k_l= b^k_{l+1}-b^k_l \clm^{k-l}(P)\notag\\
=\,&b^{k+1}_{l+1}-b^{k+1}_l \clm^{k+1-l}(Q)
=a^{k+1}_{l+1}-Q\clm^{-1}(a^{k+1}_l).\label{re-coef}\\
&c^k_l-Q\clm^{-1}(c^k_{l+1})=d^k_l-\clm^{-k+l+1}(Q) d^k_{l+1}\notag\\
=\,&d^{k+1}_l-\clm^{-k+l}(P) d^{k+1}_{l+1}
=\clm(c^{k+1}_l)-P c^{k+1}_{l+1}.\label{re-coef-m}
\end{align}
Starting from the initial conditions
\begin{align*}
&a^k_0=b_0^k=1,\quad a_l^0=b^0_l=0,\quad k\ge 0,\, l\ge 1,\\
&c^0_0=d^0_0=1,\quad c_l^0=d^0_l=0,\quad l\ge 1,\\
&c^k_0=\prod_{i=0}^{k-1} \clm^{-i}\Bigl(\frac{Q}{P}\Bigr),
\quad d^k_0=\prod_{i=0}^{k-1} \clm^{-i}\Bigl(\frac{Q}{P^-}\Bigr),
\quad c_l^0=d^0_l=0,\quad k\ge 1,
\end{align*}
we can obtain $a^k_l, b^k_l, c^k_l, d^k_l$ uniquely from \eqref{re-coef} and \eqref{re-coef-m}. For example, we have
\begin{align*}
&a^1_1=Q-P, \quad b^1_1=Q^+-P,\\
&a^2_1=Q^++Q-P^+-P,\quad b^2_1=Q^{++}+Q^+-P^+-P,\\
&a^1_2=Q(Q^--P^-),\quad b^1_2=Q(Q^+-P),\\
&a^2_2=Q(Q+Q^--P-P^-)+Q Q^+-P Q^+-PQ+P^2,\\
&b^2_2=Q^+(Q^{++}+Q^+-P^+-P)+Q Q^+-PQ^+-PQ+P^2.
\end{align*}
and
\begin{align*}
&c^1_1=\frac{Q^+}{P P^+}-\frac1{P}, \quad d^1_1=\frac{Q}{P P^-}-\frac1{P},\\
&c^2_1=\frac1{P}\left(\frac{QQ^+}{PP^+}-\frac{Q}{P}+\frac{Q^2}{PP^-}-\frac{Q}{P^-}\right),\\
&d^2_1=\frac1{P^-}\left(\frac{QQ^-}{P^-P^{--}}-\frac{Q}{P^-}+\frac{Q^2}{PP^-}-\frac{Q}{P}\right),\\
&c^1_2=\frac1{P}\left(\frac{Q^{++}}{P^+ P^{++}}-\frac1{P^+}\right),\quad d^1_2=\frac1{P^+}\left(\frac{Q}{P P^-}-\frac1{P}\right).
\end{align*}
\begin{ex}
The first and the second positive flows of the Ablowitz-Ladik hierarchy have the expressions
\begin{align}
P_{t_0}&=P(Q^+-Q),\quad Q_{t_0}=Q (Q^+-Q^--P+P^-).\label{exm-1}\\
P_{t_1}&=\frac12 P(PQ-PQ^++P^-Q-P^+Q^++Q^+Q^{++}+Q^+Q^+-QQ^--Q^2),\label{exm-2}\\
Q_{t_1}&=\frac12Q(P^2-P^-P^--P^+Q^+-2PQ^+-PQ+P^-Q+2P^-Q^-+P^{--}Q^-\label{exm-3}\\
&\quad +Q^+Q^{++}+Q^+Q^++QQ^+-Q^-Q^{--}-QQ^--Q^-Q^-),\label{exm-4}
\end{align}
and the negative flows read
\begin{align}
P_{s_0}&=\frac{Q^+}{P^+}-\frac{Q}{P^-},\quad Q_{s_0}=\frac{Q}{P}-\frac{Q}{P^-},\label{exm-5}\\
P_{s_1}&=\frac1{2P^+}\left(\frac{Q^+Q^{++}}{P^+P^{++}}-\frac{Q^+}{P^+}+\frac{Q^+Q^+}{PP^+}-\frac{Q^+}{P}\right)\\
&\quad-\frac1{2P^-}\left(\frac{QQ^-}{P^-P^{--}}-\frac{Q}{P^-}+\frac{Q^2}{PP^-}-\frac{Q}{P}\right),\label{exm-6}\\
Q_{s_1}&=\frac1{2P}\left(\frac{QQ^+}{PP^+}-\frac{Q}{P}+\frac{Q^2}{PP^-}\right)
-\frac1{2P^-}\left(\frac{QQ^-}{P^-P^{--}}-\frac{Q}{P^-}+\frac{Q^2}{PP^-}\right).\label{exm-7}
\end{align}
\end{ex}

\section{Hamiltonian structure of the Ablowitz-Ladik hierarchy}
In this section, we give the Hamiltonian structure of the Ablowitz-Ladik hierarchy and derive a recursion relation among its flows.
This recursion relation leads to a tri-Hamiltonian structure of the Ablowitz-Ladik hierarchy, as we will show in the next section.

\begin{thm}\label{thm-1}
The flows of the Ablowitz-Ladik hierarchy can be represented by the following Hamiltonian systems:
\beq\label{al-ham-pq}
\begin{pmatrix} P_{t_k}\\ Q_{t_k}\end{pmatrix}=\mathcal{P}_1\begin{pmatrix} \frac{\delta H_k}{\delta P} \vsp\\ \frac{\delta H_k}{\delta Q}\end{pmatrix},\quad
\begin{pmatrix} P_{s_k}\\ Q_{s_k}\end{pmatrix}=\mathcal{P}_1\begin{pmatrix} \frac{\delta G_k}{\delta P}\vsp\\ \frac{\delta G_k}{\delta Q}\end{pmatrix},\quad k\ge 0,
\eeq
where the Hamiltonian operator ${\mathcal P}_1$ is given by
\beq\label{tpo-1}
\mathcal{P}_1
=\begin{pmatrix} Q\clm^{-1}-\clm Q& (1-\clm)Q\vsp\\ Q(\clm^{-1}-1)& 0\end{pmatrix},
\eeq
and the Hamiltonians have the expressions
\[H_k=\sum_{n\in\mathbb{Z}} h_k(n), \quad \textrm{with}\quad h_k=\frac1{(k+2)!}\res L^{k+2},\quad k\ge -1,\]
and
\[G_0=-\sum_{n\in\mathbb{Z}} \log P(n), \ G_k=\sum_{n\in\mathbb{Z}} g_k(n), \ \textrm{with}\ g_k=\frac1{k(k+1)!}\res M^k,\quad k\ge 1.\]
Here, for an operator $K=\sum_{j\in\mathbb{Z}} a_j \clm^j$ we define
$\res K=a_0$, and the variational derivatives of a functional $H=\sum_{n\in\mathbb{Z}} h(n)$ are defined by
\[\frac{\delta H}{\delta P(n)}=\sum_{l\in\mathbb{Z}} \clm^{-l}\left(\frac{\p h(n)}{\p P(n+l)}\right),\quad \frac{\delta H}{\delta Q(n)}=\sum_{l\in\mathbb{Z}} \clm^{-l}\left(\frac{\p h(n)}{\p Q(n+l)}\right).\]
\end{thm}
\begin{proof}
We first note that the operator $\mathcal{P}_1$ is a Hamiltonian operator, since if we introduce the unknown functions
\begin{equation}\label{vpq}
v^1=Q-P, \quad v^2=\log Q
\end{equation}
and denote
\[J=\begin{pmatrix}\frac{\p u}{\p P} &\frac{\p u}{\p Q} \vvsp\\
\frac{\p v}{\p P} &\frac{\p v}{\p Q}\end{pmatrix}
=\begin{pmatrix}-1 &1 \\
0 &\frac1{Q}\end{pmatrix},
\]
then in terms of the new unknown functions we have the constant Hamiltonian operator
\[J\mathcal{P}_1J^T=\begin{pmatrix}0 & \clm-1 \vsp\\
1-\clm^{-1}&0\end{pmatrix}.\]

To derive the Hamiltonian formalism \eqref{al-ham-pq} of the Ablowitz-Ladik hierarchy, let us introduce the equivalence relation $\sim$ on the space of
differential 1-forms
\[\omega=\sum_{l\in\mathbb{Z}} \left(f_l \delta P^{(l)}+g_l \delta Q^{(l)}\right),\]
where $f_l, g_l$ are smooth functions of
$P^{(k)}(n)=P(n+k), Q^{(k)}(n)=Q(n+k)$ $(k\in\mathbb{Z})$. We say that two 1-forms are equivalent if their difference can be represented as the action of $\clm-1$ on another 1-form.
Then we have
\begin{align*}
\delta h_k&=\sum_{l\in\mathbb{Z}}\left(\frac{\p h_k}{\p P^{(l)}} \delta P^{(l)}+\frac{\p h_k}{\p Q^{(l)}} \delta Q^{(l)}\right)\\
&\sim \frac{\delta H_k}{\delta P}\delta P+\frac{\delta H_k}{\delta Q}\delta Q.
\end{align*}
On the other hand, from the definition of $h_k$ we have
\begin{align*}
\delta h_k&\sim \frac1{(k+1)!}\res L^{k+1}\delta L
=\frac1{(k+1)!}\res L^{k+1}\left(B^{-1}\delta Q \clm^{-1}L-B^{-1}\delta P\right)\\
&\sim\frac1{(k+1)!}\res\left(\delta Q \clm^{-1} L^{k+2} B^{-1}-\delta P L^{k+1} B^{-1}\right).
\end{align*}
From the above relation it follows that
\begin{align}
\frac{\delta H_k}{\delta P}=-\frac1{(k+1)!}\res\left(L^{k+1} B^{-1}\right),\quad
\frac{\delta H_k}{\delta Q}=\frac1{(k+1)!}\res\left(\clm^{-1} L^{k+2} B^{-1}\right).\label{vH}
\end{align}
In order to prove that the positive flows of the Ablowitz-Ladik hierarchy \eqref{al-def} has the Hamiltonian representation given in \eqref{al-ham-pq}, we need to verify that
\begin{align*}
&\frac1{(k+1)!}\left((\tL^{k+1})_+A-A (L^{k+1})_+\right)=\left(\clm Q-Q\clm^{-1}\right)\left(\frac{\delta H_k}{\delta P}\right)
+\left(\clm-1\right)Q\left(\frac{\delta H_k}{\delta Q}\right),\\
&\frac1{(k+1)!}\left((\tL^{k+1})_+B-B (L^{k+1})_+\right)
=Q(1-\clm^{-1}) \left(\frac{\delta H_k}{\delta P}\right)\clm^{-1}.
\end{align*}
\begin{rmk}
For a difference (or differential) operator $\mathcal{W}$ and a function $f$, we use the notation $\mathcal{W}(f)$ to denote the function obtained by the action of $\mathcal{W}$ on $f$. We understand the expression $\mathcal{W} f$ as the product of two operators $\mathcal{W}$ and $f$ (view it as an operator).
\end{rmk}
\noindent By using the idenities
\begin{align*}
&(\tL^{k+1})_+A-A (L^{k+1})_+
=\res\left(\clm L^{k+1}\right)
-\res\left(\tL^{k+1}\clm\right),\\
&(\tL^{k+1})_+B-B (L^{k+1})_+
=Q\clm^{-1}\res\left(L^{k+1}\right)
-\res\left(\tL^{k+1}\right)Q\clm^{-1}
\end{align*}
and \eqref{vH} we are led to prove the identities
\begin{align}
&\res\left(\clm L^{k+1}\right)-\res\left(\tL^{k+1}\clm\right)\notag\\
=\,&-(\clm Q-Q\clm^{-1})\left(\res\left(L^{k+1}B^{-1}\right)\right)+(\clm-1)Q\left(\res\left(\clm^{-1} L^{k+2} B^{-1}\right)\right),\label{iden-1}\\
\,&\clm \left(\res\left(\tL^{k+1}\right)\right)-\res\left(L^{k+1}\right)\notag\\
=\,&(\clm-1)\left(\res\left(L^{k+1}B^{-1}\right)\right).\label{iden-2}
\end{align}

\begin{lem}
The identities \eqref{iden-1}, \eqref{iden-2} hold true.
\end{lem}
\begin{proof}
By using the definition \eqref{def-AB}, we can rewrite the left hand side of \eqref{iden-2} as follows:
\begin{align*}
&\clm \left(\res\left(B B^{-1}\tL^{k+1}\right)\right)-\res\left(L^{k+1}B^{-1}B\right)\\
=\,&\clm\left(\res\left(L^{k+1}B^{-1}\right)\right)-\clm Q\left(\res\left(\clm^{-1}L^{k+1}B^{-1}\right)\right)\\
&\quad-\res\left(L^{k+1}B^{-1}\right)+\res\left(L^{k+1}B^{-1}Q\clm^{-1}\right)\\
=\,&(\clm-1)\left(\res\left(L^{k+1}B^{-1}\right)\right),
\end{align*}
so the identity \eqref{iden-2} holds true. In a similar way, we can rewrite the left hand side of \eqref{iden-1} as follows:
\begin{align}
&\res\left(\clm L^{k+1}B^{-1}B\right)-\res\left(BB^{-1}\tL^{k+1}\clm\right)\notag\\
=\,&\left[\res\left(\clm L^{k+1}B^{-1}\right)-\clm Q\left(\res\left(L^{k+1}B^{-1}\right)\right)\right]\notag\\
&\quad-\left[\res\left(L^{k+1}B^{-1}\clm\right)-Q\clm^{-1}\left(\res\left(L^{k+1}B^{-1}\right)\right)\right]\notag\\
=\,&-(\clm Q-Q\clm^{-1})\left(\res\left(L^{k+1}B^{-1}\right)\right)
+\res\left(AL^{k+1}B^{-1}\right)-\res\left(L^{k+1}B^{-1}A\right)\notag\\
=\,&-(\clm Q-Q\clm^{-1})\res\left(L^{k+1}B^{-1}\right)
+\res\left(\tL^{k+2}\right)-\res\left(L^{k+2}\right).\label{iden-1-a}
\end{align}
On the other hand, by using the identity \eqref{iden-2}
we have
\begin{align}
&(\clm-1)Q\left(\res\left(\clm^{-1} L^{k+2} B^{-1}\right)\right)\notag\\
=\,&-(\clm-1)\left(\res\left(BL^{k+2} B^{-1}\right)\right)+(\clm-1)\left(\res\left(L^{k+2} B^{-1}\right)\right)\notag\\
=\,&-(\clm-1)\left(\res\left(\tL^{k+2}\right)\right)
+\clm \left(\res\left(\tL^{k+2}\right)\right)-\res\left(L^{k+2}\right)\notag\\
=\,&\res\left(\tL^{k+2}\right)
-\res\left(L^{k+2}\right).\label{iden-1-b}
\end{align}
Then the identity \eqref{iden-1} follows from the relations
\eqref{iden-1-a} and \eqref{iden-1-b}. The lemma is proved.
\end{proof}
We thus proved that the positive flows of the Ablowitz-Ladik hierarchy
have the Hamiltonian formalism given in \eqref{al-ham-pq}. We can prove in a similar way the Hamiltonian formalism of the negative flows of the Ablowitz-Ladik hierarchy. The theorem is proved.
\end{proof}

\begin{rmk}
The Hamiltonian structure of the Ablowitz-Ladik hierarchy given in the above theorem was obtained in \cite{Brini-2, Brini-3} by using the
Hamiltonian structure of the 2D Toda lattice hierarchy. Here we give a
direct proof of this result.
\end{rmk}

Now let us introduce the operator
\beq\label{tpo-2}
\mathcal P_2=\begin{pmatrix}0 & P(\clm-1) Q\vsp\\ Q(1-\clm^{-1}) P & Q(\clm-\clm^{-1}) Q\end{pmatrix}.
\eeq
Then we have the following theorem.

\begin{thm}\label{thm-2}
The Ablowitz-Ladik hierarchy can also be represented in the following form:
\beq\label{al-ham-pq-2}
\begin{pmatrix} P_{t_k}\\ Q_{t_k}\end{pmatrix}=\frac1{k+1}\mathcal{P}_2\begin{pmatrix} \frac{\delta H_{k-1}}{\delta P}\vsp\\ \frac{\delta H_{k-1}}{\delta Q}\end{pmatrix},
\quad \begin{pmatrix} P_{s_k}\\ Q_{s_k}\end{pmatrix}=(k+2)\mathcal{P}_2\begin{pmatrix} \frac{\delta G_{k+1}}{\delta P}\vsp\\ \frac{\delta G_{k+1}}{\delta Q}\end{pmatrix},
\quad k\ge 0,
\eeq
\end{thm}
\begin{proof}
From the definition \eqref{eq-A}, \eqref{eq-B} of the positive flows of the Ablowitz-Ladik hierarchy and the formulae of the variational derivatives given in \eqref{vH} it follows that, in order to prove the representation of the positive flows given in \eqref{al-ham-pq-2}, we only need to prove the following identities:
\begin{align}
&\res\left(\clm L^{k+1} \right)-\res\left(\tL^{k+1}\clm \right)\notag\\
=\,&P(1-\clm)Q \left(\res\left(\clm^{-1} L^{k+1}B^{-1}\right)\right),\label{iden-2-1}\\
&\res\left(\tL^{k+1}\right)-\clm^{-1}\left(\res\left(L^{k+1}\right)\right)\notag\\
=\,&(\clm^{-1}-1) P\left(\res\left(L^k B^{-1}\right)\right)+(\clm-\clm^{-1})Q\left(\res\left(\clm^{-1}L^{k+1} B^{-1}\right)\right).\label{iden-2-2}
\end{align}
By using the definition of the operators $L, \tL$  we have
\begin{align*}
&\res\left(\clm L^{k+1} \right)-\res\left(\tL^{k+1}\clm \right)\\
=\,&\res\left((A+P) L^{k+1} \right)-\res\left(\tL^{k+1}(A+P) \right)\\
=\,&P\res\left(L^{k+1} \right)-P\res\left(\tL^{k+1}\right)\\
=\,&P\res\left(L^{k+1}B^{-1} B\right)-P\res\left(B B^{-1}\tL^{k+1}\right)\\
=\,&-P\res\left(L^{k+1}B^{-1} Q\clm^{-1}\right)+P\res\left(Q\clm^{-1} B^{-1}\tL^{k+1}\right)\\
=\,&-P\clm Q\left(\res\left(\clm^{-1}L^{k+1}B^{-1}\right)\right)+PQ\left(\res\left(\clm^{-1}L^{k+1}B^{-1}\right)\right)\\
=\,&P(1-\clm)Q \left(\res\left(\clm^{-1} L^{k+1}B^{-1}\right)\right),
\end{align*}
so the identity \eqref{iden-2-1} holds true.

To prove the identity \eqref{iden-2-2}, we first first note that its right hand side can be written as
\begin{align}
&(\clm^{-1}-1) \left(\res\left(L^k B^{-1}\clm\right)-\res\left(L^{k+1}\right)\right)\notag\\
&\quad +(\clm-\clm^{-1})\left(\res\left(L^{k+1} B^{-1}\right)-\res\left(\tL^{k+1}\right)\right).\label{zh-2-1}
\end{align}
By using the identity \eqref{iden-2} and the fact that
\begin{align*}
&(\clm^{-1}-1)\left(\res\left(L^k B^{-1}\clm\right)\right)\\
=\,&\clm^{-1}\left(\res\left(L^k B^{-1}\clm\right)\right)-\clm^{-1}\left(\res\left(\clm L^k B^{-1}\right)\right)\\
=\,&\clm^{-1}\left(\res\left(L^k B^{-1}(\clm-P)\right)\right)-\clm^{-1}\left(\res\left((\clm-P) L^k B^{-1}\right)\right)\\
=\,&\clm^{-1}\left(\res\left(L^{k+1}\right)\right)-\clm^{-1}\left(\res\left(\tL^{k+1} \right)\right),
\end{align*}
we can represent \eqref{zh-2-1} in the form
\begin{align*}
&\clm^{-1}\left(\res\left(L^{k+1}\right)\right)-\clm^{-1}\left(\res\left(\tL^{k+1} \right)\right)
-(\clm^{-1}-1)\left(\res\left(L^{k+1}\right)\right)\\
&\quad+(\clm+1)\left(\res\left(\tL^{k+1}\right)\right)-(1+\clm^{-1})\left(\res\left(L^{k+1}\right)\right)\\
&\quad -(\clm-\clm^{-1})\left(\res\left(\tL^{k+1}\right)\right)\\
=\,&\res\left(\tL^{k+1}\right)-\clm^{-1}\left(\res\left(L^{k+1}\right)\right).
\end{align*}
Thus the identity \eqref{iden-2-2} also holds true, and the representation \eqref{al-ham-pq-2} for the positive flows of the Ablowitz-Ladik hierarchy holds true. The proof of \eqref{al-ham-pq-2} for the negative flows is similar, so we omit it here. The theorem is proved.
\end{proof}

Define the operator
\beq\label{recur-op}
\mathcal{R}=\mathcal P_2\circ \mathcal P_1^{-1},
\eeq
then from \eqref{tpo-1} and \eqref{tpo-2} we can obtain its explicit expression
\[
\mathcal{R}=\left(
\begin{array}{cc}
-P & P\left(Q-\Lambda Q \Lambda\right)\left(1-\Lambda\right)^{-1}Q^{-1} \vsp\\
-Q\left(1+\Lambda^{-1}\right) & Q\left(\left(\Lambda-1\right)P^-+\left(1+\Lambda^{-1}\right)\left(Q-\Lambda Q \Lambda\right)\right)\left(1-\Lambda\right)^{-1}Q^{-1}
\end{array}
\right).
\]

By using Theorem \ref{thm-1} and Theorem \ref{thm-2} we arrive at the following proposition.

\begin{prop}
The Ablowitz-Ladik hierarchy satisfies the following recursion relation:
\beq
\begin{pmatrix} P_{t_k}\\ Q_{t_k}\end{pmatrix}=\frac1{k+1}\mathcal{R}\begin{pmatrix} P_{t_{k-1}}\\ Q_{t_{k-1}}\end{pmatrix},\quad
\begin{pmatrix} P_{s_k}\\ Q_{s_k}\end{pmatrix}=\frac1{k+1}\mathcal{R}^{-1}\begin{pmatrix} P_{s_{k-1}}\\ Q_{s_{k-1}}\end{pmatrix},
\quad k\ge 1
\eeq
with the recursion operator $\mathcal{R}$ defined by \eqref{recur-op}.
\end{prop}

By using the recursion operator we can also define the following operator
\beq\label{tpo-3}
\mathcal P_3=\mathcal{R}\circ\mathcal P_2,
\eeq
it has the following expression:
\begin{align}\label{operator-K}
\mathcal P_3=&\left(
\begin{array}{cc}
P\left(Q\Lambda^{-1}-\Lambda Q\right)P & K_{12} \vsp\\
K_{21} & K_{22}
\end{array}
\right),
\end{align}
where
\begin{align*}
K_{12}=\,& P\left(\left(Q-\Lambda Q \Lambda\right)\left(1+\Lambda^{-1}\right)-P\left(1-\Lambda\right)\right)Q, \\
K_{21}=\,& Q\left(\left(\Lambda+1\right)\left(\Lambda^{-1}Q\Lambda^{-1}-Q\right)+\left(1-\Lambda^{-1}\right)P\right)P, \\
K_{22}=\,& Q\left(1+\Lambda^{-1}\right)\left(Q-\Lambda Q \Lambda\right)\left(1+\Lambda^{-1}\right)Q \\
\,&+Q\left(1+\Lambda^{-1}\right)P\left(\Lambda-1\right)Q+Q\left(\Lambda-1\right)P^-\left(1+\Lambda^{-1}\right)Q
\\
=\,&
  Q\left(
    \left(1+\Lambda^{-1}\right)\left(Q-\Lambda Q \Lambda\right)\left(1+\Lambda^{-1}\right)
   +2(P\Lmd-\Lmd^{-1}P)
  \right)Q.
\end{align*}
Then the Ablowitz-Ladik hierarchy can also be represented in the form
\begin{align}\label{al-ham-pq-3}
&\begin{pmatrix} P_{t_0}\\ Q_{t_0}\end{pmatrix}=\mathcal{P}_3\begin{pmatrix} \frac{\delta G_0}{\delta P}\vsp\\ \frac{\delta G_0}{\delta Q}\end{pmatrix},
\quad \begin{pmatrix} P_{t_k}\\ Q_{t_k}\end{pmatrix}=\frac1{k(k+1)}\mathcal{P}_3\begin{pmatrix} \frac{\delta H_{k-2}}{\delta P}\vsp\\ \frac{\delta H_{k-2}}{\delta Q}\end{pmatrix},
\quad k\ge 1;
\\
&\begin{pmatrix} P_{s_k}\\ Q_{s_k}\end{pmatrix}=(k+2)(k+3)\mathcal{P}_3\begin{pmatrix} \frac{\delta G_{k+2}}{\delta P}\vsp\\ \frac{\delta G_{k+2}}{\delta Q}\end{pmatrix},\quad k\ge 0.
\end{align}

In the next section, we are to show that the operators $\mathcal{P}_2$, $\mathcal{P}_3$ are Hamiltonian operators, and the Hamiltonian operators $\mathcal{P}_1$, $\mathcal{P}_2$, $\mathcal{P}_3$ are compatible. Thus the
Ablowitz-Ladik hierarchy is a hierarchy of tri-Hamiltonian systems.

\section{A tri-Hamiltonian structure of the Ablowitz-Ladik hierarchy}
We are to use the notion of Schouten bracket defined on the space of local functionals of a super manifold
$\hat{M}$ to show that $\mathcal P_1, \mathcal P_2, \mathcal P_3$
form a tri-Hamiltonian structure, here $M$ is a smooth manifold of dimension $m$, and $\hat{M}$ is obtained from the cotangent bundle of $M$ with its fiber's parity reversed, see \cite{LZ-2011, LZ-2013} for details.

Let $\hat{U}=U\times \mathbb{R}^{0|m}$ be a local trivialization of $\hat{M}$, $u^1,\dots, u^m$ be coordinates on $U$, and $\theta_1,\dots,\theta_m$ be the dual coordinates (or the super variables) on $\mathbb{R}^{0|m}$. Then on the infinite jet space $J^\infty(\hat M)$ we have local
coordinates
$\{u^{\al, s},\theta_\al^s\mid \al=1,\dots,m, s\ge 0\}$ with $u^{\al,0}=u^\al$, $\theta_\al^0=\theta_\al$. The ring of differential polynomials
$\hat{\mathcal{A}}$ is locally given by
\[C^\infty(\hat{U})[[u^{\al,s}, \theta_\al^s\mid \al=1,\dots,m, s\ge 0]],\]
and the space of local functionals of $\hat{M}$ is defined by the quotient space
\[\hat{\mathcal{F}}=\hat{\mathcal{A}}/\p \hat{\mathcal{A}},\]
where $\p$ is the vector field on $J^\infty(\hat M)$ given by
\[\p=\sum_{s\ge 0}\left(u^{\al,s+1}\frac{\p}{\p u^{\al,s}}+\theta_\al^{s+1}\frac{\p}{\p\theta_\al^s}\right).\]
We denote elements of $\hat{\mathcal{F}}$ by $\int f\nd x$ with $f\in \hat{\mathcal{A}}$.

Now we can define the Schouten bracket on the space of local functionals of $\hat{M}$ as follows:
\[[H, G]=\int \left(\frac{\delta H}{\delta \theta_\al}\frac{\delta G}{\delta u^\al}+(-1)^p\frac{\delta H}{\delta u^\al}\frac{\delta G}{\delta \theta_\al}\right)\nd x\]
with the variational derivatives
\[\frac{\delta H}{\delta u^\al}=\sum_{s\ge 0}(-\p)^s\frac{\p h}{\p u^{\al,s}},\quad \frac{\delta H}{\delta \theta_\al}=\sum_{s\ge 0}(-\p)^s\frac{\p h}{\p\theta_\al^s}.\]
Here we assume that the super degree of $H$ is $p$ which is defined by
\[\deg \theta_\al^s=s, \quad \deg u^{\al,s}=\deg f=0\]
for any $f\in C^\infty(M)$.
Then an operator
\[\mathcal{P}=\Biggl(\sum_{s\ge 0} P_s^{\al\beta}({\mathbf u}, {\mathbf u}_x,\dots) \p_x^s\Biggr)\]
is a Hamiltonian operator iff it is anti-symmetric and the associated local functional
\[I=\frac12 \int \theta_\al \mathcal{P}^{\al\beta}_s\p^s \theta_\beta\nd x\]
satisfies the condition $[I, I]=0$.

In order to adapt the above setting to the discrete case,
we introduce the variables
$u^1(x), u^2(x)$ by applying an $\ve$-interpolation to the discrete variables $P(n), Q(n)$ so that
\[P(n):=u^1(x)|_{x=n\ve},\quad Q(n)=u^2(x)|_{x=n\ve}.\]
Then the shift operator $\clm$ can be represented as $\clm=e^{\ve\p_x}$, i.e.
\[\clm \left(u^\al\right)(x)=u^\al(x+\ve),\quad \al=1,2,\]
and the operators $\mathcal P_1$, $\mathcal P_2$ defined in \eqref{tpo-1} and
\eqref{tpo-2} can be rewritten as
\begin{align}
&\mathcal{P}_1
=\begin{pmatrix} u^2(x) e^{-\ve\p_x}-e^{\ve\p_x} u^2(x)& (1-e^{\ve\p_x})u^2(x)\vsp\\ u^2(x)(e^{-\ve\p_x}-1)& 0\end{pmatrix},\label{tpo-1-x}\\
&\mathcal{P}_2
=\begin{pmatrix} 0& u^1(x)(e^{\ve\p_x}-1)u^2(x)\vsp\\ u^2(x)(1-e^{-\ve\p_x})u^1(x)& u^2(x)(e^{\ve\p_x}-e^{-\ve\p_x}) u^2(x)\end{pmatrix}.\label{tpo-2-x}
\end{align}
The operator $\mathcal P_3$ has a similar expressions obtained by the substitution
\[P\to u^1,\ Q\to u^2, \ \clm\to e^{\ve\p_x}\]
in \eqref{operator-K}.

It is easy to see that if the density $h$ of the local functionals $H=\int h\nd x$
can be represented as a function depend only on $u^\al, \clm^s(u^\al),
\theta_\al, \clm^s(\theta_\al)$ with $\al=1,\dots,m$ and $s\in\mathbb{Z}$, then its variational derivatives can be represented as
\[\frac{\delta H}{\delta u^\al}=\sum_{s\in\mathbb{Z}}\clm^{-s}\left(\frac{\p h}{\p \clm^s(u^{\al})}\right),\quad \frac{\delta H}{\delta \theta_\al}=\sum_{s\in\mathbb{Z}}\clm^{-s}\left(\frac{\p h}{\p\clm^s(\theta_\al)}\right).\]

To simplify our notation, in what follows we will also denote $u^1, u^2$ by $u_1, u_2$ respectively and
\[u_\al^\pm=\clm^{\pm 1} \left(u^\al\right),\quad \theta_\al^{\pm}=\clm^{\pm 1}\left(\theta_\al\right),\quad \al=1,2.\]

\begin{ex}
The local functional $I$ corresponding to the operator $\mathcal{P}_1$ has the expression
\begin{align*}
  I=&\frac12\int
  (\theta_1,\theta_2)
  \begin{pmatrix} u_2\clm^{-1}-\clm u_2& (1-\clm)u_2\\ u_2(\clm^{-1}-1)& 0\end{pmatrix}
  \begin{pmatrix}
    \theta_1 \\
    \theta_2
  \end{pmatrix}\nd x\\
=&
  \frac12\int\left(
   -u_2^+\theta_1\theta_1^+
   +u_2\theta_1\theta_1^-
   -u_2^+\theta_1\theta_2^+
   +2u_2\theta_1\theta_2
   -u_2\theta_1^-\theta_2
  \right)\nd x.
\end{align*}
Its variational derivatives are given by
\begin{align*}
\dtafrac{I}{\theta_1}
&=
   u_2\theta_1^--u_2^+\theta_1^++u_2\theta_2-u_2^+\theta_2^+
\\
 &=(u_2\clm^{-1}-\clm u_2)\left(\theta_1\right)+(1-\clm)(u_2\theta_2),\\
\dtafrac{I}{\theta_2}
 &=u_2(\theta_1^--\theta_1)=u_2(\clm^{-1}-1)\left(\theta_1\right),\\
\dtafrac{I}{u_1}&=0,\\
\dtafrac{I}{u_2}&=\theta_1\theta_1^--\theta_1^-\theta_2+\theta_1\theta_2.
\end{align*}
By using that fact that $\theta_\al\theta_\al=\theta_\al^-\theta_\al^-=0$,
$\theta^-_\afa\theta_\afa+\theta_\afa\theta^-_\afa=0$,
we can calculate the Schouten bracket $[I,I]$ as follows:
\begin{align*}
  [I,I]
&=
  2\int
  \left(
   \dtafrac{I}{\theta_1}
   \dtafrac{I}{u_1}
  +\dtafrac{I}{\theta_2}
   \dtafrac{I}{u_2}
  \right)\td x
\\
&=
  2\int u_2(\theta_1^--\theta_1)
  (\theta_1\theta_1^--\theta_1^-\theta_2+\theta_1\theta_1)\td x
\\
&=
 2\int u_2(\theta_1^-\theta_1\theta_2+\theta_1\theta_1^-\theta_2)\td x=0.
\end{align*}
Thus we know that $\mcalP_1$ is a Hamiltonian operator. Note that we already showed this fact in the last section by applying a change of variables
to the operator $\mathcal{P}_1$.
\end{ex}

Now we use this method to show
that the operators $\mathcal{P}_2$, $\mathcal{P}_3$
in \eqref{tpo-2}, \eqref{tpo-3}
are Hamiltonian operators, and the Hamiltonian operators $\mathcal{P}_1$, $\mathcal{P}_2$, $\mathcal{P}_3$ are compatible.
Let $J, K$ be the local functionals corresponding to
$\mcalP_2, \mcalP_3$ respectively, i.e.
\begin{equation}\label{J}
\begin{split}
  J
&=\frac12\int
  ({\theta_1},{\theta_2})
  \begin{pmatrix}
    0 & {u_1}(\clm-1){u_2} \vsp\\
    {u_2}(1-\clm^{-1}){u_1} & {u_2}(\clm-\clm^{-1}){u_2}
  \end{pmatrix}
  \begin{pmatrix}
    {\theta_1} \\
    {\theta_2}
  \end{pmatrix}\nd x\\
&=
  \frac12\int
  \left(
   -2{u_1}{u_2}\theta_1\theta_2
  +u_1^-{u_2}\theta_1^-\theta_2
  +{u_1}u_2^+\theta_1\theta_2^+
  -{u_2}u_2^-\theta_2\theta_2^-
  +{u_2}u_2^+\theta_2\theta_2^+
  \right)\nd x,
\end{split}
\end{equation}
and 
\begin{equation}\label{K}
\begin{split}
  K&=
  \frac{1}{2}
  \int({\theta_1}, {\theta_2})
  \left(
\begin{array}{cc}
{u_1}\left({u_2}\Lambda^{-1}-\Lambda {u_2}\right){u_1} & K_{12} \vsp\\
K_{21} & K_{22}
\end{array}
\right)
  \begin{pmatrix}
    {\theta_1} \\
    {\theta_2}
  \end{pmatrix}\td x\\
&=
  \int\Bigl(
    -u_1^-{u_1}{u_2}\theta_1^-\theta_1
\\
&\quad 
    +\left(
      {u_1}{u_2}{u_2}\theta_1\theta_2
     +{u_1}u_2^-{u_2}\theta_1\theta_2^-
     -{u_1}u_2^+u_2^+\theta_1\theta_2^+
     -{u_1}u_2^+u_2^{++}\theta_1\theta_2^{++}
    \right)
\\
&\quad 
   +\left(
     {u_1}{u_1}u_2^+\theta_1\theta_2^+
    -{u_1}{u_1}{u_2}\theta_1\theta_2
   \right)
   +2{u_1}{u_2}u_2^+\theta_2\theta_2^+
\\
&\quad 
   -\left(
     {u_2}u_2^+({u_2}+u_2^+)\theta_2\theta_2^+
    +{u_2}u_2^+u_2^{++}\theta_2\theta_2^{++}
   \right)
  \Bigr)\td x.
\end{split}
\end{equation}
We first show that $(\mcalP_1, \mcalP_2)$ is a bihamiltonian structure.
\begin{lem}
The operator $\mcalP_2$ defined in \eqref{tpo-2}
is a Hamiltonian operator which is compatible with the Hamiltonian operator $\mcalP_1$.
\end{lem}

\begin{proof}
  We only need to check $[J,J]=0$ and $[I,J]=0$.
The variational derivatives of $J$ are given by
\begin{align*}
  \dtafrac{J}{u_1}
&=
u_2^+\theta_1\theta_2^+-{u_2}\theta_1\theta_2
\\
&=\theta_1(\clm-1)({u_2}\theta_2),
\\
  \dtafrac{J}{{u_2}}
&=
  -{u_1}\theta_1\theta_2+{u_1}^-\theta_1^-\theta_2
  +u_2^-\theta_2^-\theta_2
  +u_2^+\theta_2\theta_2^+
\\
&=\theta_2(1-\clm^{-1})({u_1}\theta_1)
  +\theta_2(\clm-\clm^{-1})({u_2}\theta_2),
\\
  \dtafrac{J}{\theta_1}
&=
  {u_1}(\clm-1)({u_2}\theta_2),
\\
  \dtafrac{J}{\theta_2}
&=
  {u_2}(1-\clm^{-1})({u_1}\theta_1)
 +{u_2}(\clm-\clm^{-1})({u_2}\theta_2).
\end{align*}Then we have
\begin{align*}
[J,J]&=
  2\int\left(
    \dtafrac J{\theta_1}\dtafrac J{u_1}
   +\dtafrac J{\theta_2}\dtafrac J{u_2}
  \right)\nd x\\
&=
  2\int\Big(
  {u_1}(\clm-1)({u_2}\theta_2)\cdot\theta_1(\clm-1)({u_2}\theta_2)\\
&\quad  +\bigl({u_2}(1-\clm^{-1})({u_1}\theta_1)
  +{u_2}(\clm-\clm^{-1})({u_2}\theta_2)\bigr)\cdot\\
  &\qquad\qquad  \big(\theta_2(1-\clm^{-1})({u_1}\theta_1)
  +\theta_2(\clm-\clm^{-1})({u_2}\theta_2)\big)\Big)\nd x\\
&=
  2\int
    \Big(
      {u_2}(1-\clm^{-1})({u_1}\theta_1)\cdot
      \theta_2(\clm-\clm^{-1})({u_2}\theta_2)\\
   & \quad  +{u_2}(\clm-\clm^{-1})({u_2}\theta_2)\cdot\theta_2(1-\clm^{-1})({u_1}\theta_1)
    \Big)\nd x\\
&=0.
\end{align*}
Hence $\mcalP_2$ is a Hamiltonian operator.

Now we compute $[I,J]$.
This calculation is straightforward, but a little complicated. We have
\begin{align*}
  [I,J]
&=
  \int\left(
    \dtafrac{I}{\theta_1}
    \dtafrac{J}{u_1}
   +\dtafrac{I}{\theta_2}
    \dtafrac{J}{u_2}
   +\dtafrac{I}{u_1}
    \dtafrac{J}{\theta_1}
   +\dtafrac{I}{u_2}
    \dtafrac{J}{\theta_2}
  \right)\nd x\\
&=
 \int\Big[
   \big(({u_2}\clm^{-1}-\clm {u_2})\theta_1+(1-\clm)({u_2}\theta_2)\big)\cdot
   \theta_1(\clm-1)({u_2}\theta_2)\\
&\quad+
   {u_2}(\clm^{-1}-1)\theta_1\cdot
   \big(\theta_2(1-\clm^{-1})({u_1}\theta_1)
  +\theta_2(\clm-\clm^{-1})({u_2}\theta_2)\big)\\
&\quad +
  (\theta_1\theta_1^--\theta_1^-\theta_2+\theta_1\theta_2)\cdot
  \big({u_2}(1-\clm^{-1})({u_1}\theta_1)
 +{u_2}(\clm-\clm^{-1})({u_2}\theta_2)\big)
 \Big]\nd x\\
&=
  \int\Big[
    ({u_2}\theta_1^--u_2^+\theta_1^+)\theta_1(u_2^+\theta_2^+-{u_2}\theta_2)\\
&\quad  +{u_2}(\theta_1^--\theta_1)\theta_2({u_1}\theta_1-u_1^-\theta_1^-+u_2^+\theta_2^+-u_2^-\theta_2^-)\\
&\quad  +{u_2}(\theta_1\theta_1^--\theta_1^-\theta_2+\theta_1\theta_2)\cdot
     ({u_1}\theta_1-u_1^-\theta_1^-+u_2^+\theta_2^+-u_2^-\theta_2^-)
  \Big]\nd x\\
&=
\int\Big[
  \left(
    u_2^+u_2^+\theta_1\theta_1^+\theta_2^+
   -u_2{u_2}\theta_1^-\theta_1\theta_2
  \right)
 +\left(
    u_2^-u_2\theta_1^-\theta_1\theta_2^--{u_2}u_2^+\theta_1\theta_1^+\theta_2
  \right)
\Big]\td x\\
&=
  \int(\Lmd-1)
  \left(
    {u_2}{u_2}\theta_1^-\theta_1\theta_2
   -u_2^-{u_2}\theta_1^-\theta_1\theta_2^-
  \right)\td x\\
&=0.
\end{align*}
Hence $\mcalP_1, \mcalP_2$ are compatible.
\end{proof}

As for the third operator $\mcalP_3$ defined in \eqref{operator-K},
we need to verify that the corresponding local functional $K$ given in \eqref{K} satisfies the relations
\beq \label{triHamil-K}
[K,I]=[K,J]=[K,K]=0.
\eeq

\begin{lem}
  The relations given in \eqref{triHamil-K} hold true, i.e.
  the third operator $\mcalP_3$
  in \eqref{operator-K} is a Hamiltonian operator,
  and compatible with $\mcalP_1, \mcalP_2$.
\end{lem}
\begin{proof}
The variational derivatives of $K$ are given by
\begin{align*}
    \dtafrac{K}{u_1}
  &=
    -(u_1^-u_2\theta_1^-\theta_1 
    -u_1^+u_2^+\theta_1\theta_1^+)\\
  &\quad
    +(u_2 u_2\theta_1\theta_2
    +u_2^-u_2\theta_1\theta_2^-
    -u_2^+u_2^+\theta_1\theta_2^+
    -u_2^+u_2^{++}\theta_1\theta_2^{++})\\
  &\quad
    +2(u_1u_2^+\theta_1\theta_2^+
    -u_1u_2\theta_1\theta_2
    +u_2u_2^+\theta_2\theta_2^+), 
  \\
    \dtafrac{K}{u_2}
  &=
    -u_1^-u_1\theta_1^-\theta_1\\
  &\quad
    +(
      2u_1u_2\theta_1\theta_2
      +u_1u_2^-\theta_1\theta_2^-
      +u_1^+u_2^+\theta_1^+\theta_2
      -2u_1^-u_2\theta_1^-\theta_2
      -u_1^-u_2^+\theta_1^-\theta_2^+
      -u_1^{--}u_2^-\theta_1^{--}\theta_2
    )\\
  &\quad
    +(u_1^-u_1^-\theta_1^-\theta_2-u_1u_1\theta_1\theta_2)
    +2(u_1u_2^+\theta_2\theta_2^+
      +u_1^-u_2^-\theta_2^-\theta_2)\\
  &\quad
    -\big(
      u_2^+(2u_2+u_2^+)\theta_2\theta_2^+
     +u_2^-(u_2^-+2u_2)\theta_2^-\theta_2\\
  &\quad
     +u_2^+u_2^{++}\theta_2\theta_2^{++}
     -u_2^-u_2^+\theta_2^-\theta_2^+ 
     -u_2^{--}u_2^-\theta_2^{--}\theta_2 
    \big),
  \\
  \dtafrac{K}{\theta_1}
  &=
    (
      u_1^-u_1u_2\theta_1^-
     -u_1u_1^+u_2^+\theta_1^+ 
    )\\
  &\quad
    +(
     u_1u_2u_2\theta_2
    +u_1u_2^-u_2\theta_2^-
    -u_1u_2^+u_2^+\theta_2^+
    -u_1u_2^+u_2^{++}\theta_2^{++} 
    )\\
  &\quad
    +(
      u_1u_1u_2^+\theta_2^+
     -u_1u_1u_2\theta_2 
    ),
\\
  \dtafrac{K}{\theta_2}
&=
  (
    u_1^-u_2u_2\theta_1^-
   +u_1^{--}u_2^-u_2\theta_1^{--}
   -u_1u_2u_2\theta_1
   -u_1^+u_2u_2^+\theta_1^+ 
  )\\
&\quad
  +(u_1u_1u_2\theta_1
   -u_1^-u_1^-u_2\theta_1^-)
  +2(
    u_1u_2u_2^+\theta_2^+
   -u_1^-u_2^-u_2\theta_2^- 
  )\\
&\quad
  +\big(
    u_2^-u_2(u_2+u_2^-)\theta_2^-
   -u_2u_2^+(u_2+u_2^+)\theta_2^+
   +u_2^{--}u_2^-u_2\theta_2^{--}
   -u_2u_2^+u_2^{++}\theta_2^{++} 
  \big). 
\end{align*}  
Then it is a straightforward computation to verify the relations given in \eqref{triHamil-K}, for example, we have
\begin{align*}
  [K,I]
&=
  \int\left(
    \dtafrac{K}{\theta_1}
    \dtafrac{I}{u_1}
   +\dtafrac{K}{\theta_2}
    \dtafrac{I}{u_2}
   +\dtafrac{K}{u_1}
    \dtafrac{I}{\theta_1}
   +\dtafrac{K}{u_2}
    \dtafrac{I}{\theta_2}
  \right)\nd x\\
&=
  \int
  (1-\Lmd)
  \Big(
    u_1^{--}u_2^-u_2\theta_1^{--}\theta_1^-\theta_1       
   -u_1u_2u_2\theta_1^-\theta_1\theta_2\\
&\quad
   +2u_1^-u_2u_2\theta_1^-\theta_1\theta_2
   -2u_1^-u_2^-u_2\theta_1^-\theta_1\theta_2^-
   +u_1u_2u_2^-\theta_1\theta_1^-\theta_2^-
                                                     \\
&\quad
  +u_2^-u_2^-u_2\theta_1^-\theta_1\theta_2^-
  -u_2u_2u_2^+\theta_1^-\theta_1\theta_2^+
  +u_2^{--}u_2^-u_2\theta_1^-\theta_1\theta_2^{--}             
  -u_2u_2u_2\theta_1^-\theta_1\theta_2\\
&\quad
  +2u_2^-u_2u_2\theta_1\theta_2\theta_2^-
  -u_2u_2u_2^+\theta_1^-\theta_2\theta_2^+
  -u_2^-u_2u_2\theta_1^-\theta_2\theta_2^-
  +u_2^-u_2u_2^+\theta_1^-\theta_2^-\theta_2^+
  \Big)\td x\\
&=0.
\end{align*}
The lemma is proved.
\end{proof}

\section{Central invariants of the associated bihamiltonian structures}

In order to understand properties of the bihamiltonian integrable hierarchies associated to Gromov-Witten invariants and 2d topological field theory,
Dubrovin and the fourth-named author of the present paper proposed a program of classification of deformations bihamiltonian structures of hydrodynamic type in \cite{DZ-2001}. Together with the second-named author of the present paper, they introduced in \cite{DLZ-2006, LZ-2005} the notion of central invariants for any deformation of a given semisimple bihamiltonian structures of hydrodynamic type, and proved that two deformations are equivalent under a certain Miura-type transformation if and only if they have the same set of central invariants. In particular, all the central invariants of the deformed bihamiltonian structure of the integrable hierarchy associated to the Gromov-Witten invariants of a target space must equal to a certain constant.

In this section, we compute the central invariants of the bihamiltonian structures
$(\mathcal{P}_i, \mathcal{P}_j)$ $(1\le i\ne j\le 3)$ of the Ablowitz-Ladik hierarchy, and we find that the central invariants of the bihamiltonian structure $(\mathcal{P}_1, \mathcal{P}_2)$ equal to $\frac1{24}$, and the ones of the bihamiltonian structure $(\mathcal{P}_3, \mathcal{P}_2)$ equal to $-\frac1{24}$. This result together with properties of the leading terms of these two bihamiltonian structures implies that they are related by a certain Miura-type transformation, and it also provides a support for the validity of Brini's conjecture on the relation of the Ablowitz-Ladik hierarchy and the Gromov-Witten invariants of local $\mathbb{P}^1$.

Let us first recall the definition of the central invariants of a bihamiltonian structure with semisimple hydrodynamic limit.
Suppose $(\mcalQ_{1},\mcalQ_{2})$ is a bihamiltonian structure defined on the jet space $J^{\infty} (M)$ of a smooth manifold $M$ of dimension $m$, the components of which have the forms
\begin{align*}
\mathcal{Q}^{\al\beta}_a&=\sum_{k\ge 0} \ve^k \mathcal{Q}^{\al\beta}_{a;k}\\
&=g_a^{\afa\beta}(u)\partial_x
  +\Gamma^{\afa\beta}_{a;\gamma}(u)u^{\gamma,1}\\
&\quad   +\ve\left(\mathcal{Q}^{\al\beta}_{a;1,2}(u) \p_x^2+\mathcal{Q}^{\al\beta}_{a;1,1}(u, u_x) \p_x+\mathcal{Q}^{\al\beta}_{a;1,0}(u, u_x, u_{xx})\right)\\
&\quad +\sum_{s\ge 2}\ve^s \sum_{k=0}^{s+1} \mathcal{Q}^{\al\beta}_{a;s,k}(u, \dots, \p^{s+1-k}u) \p_x^k,\quad a=1,2,
\end{align*}
where $(u^1,\dots, u^m)$ is a local coordinate system of $M$, and $\mathcal{Q}^{\al\beta}_{a;s, k}$ are homogeneous differential polynomials of degrees $s+1-k$ with respect to the differential degree defined by $\deg\p_x^l u^j=l$.
We say that $(\mathcal{Q}^{\al\beta}_1, \mathcal{Q}^{\al\beta}_2)$ has a semisimple hydrodynamic limit, if $(g^{\al\beta})$ is non-degenerate and the characteristic equation
\[
  \det\left(
   g^{\afa\beta}_2(u)
  -\lmd g_1^{\afa\beta}(u)
  \right)=0
\]
has $n$ non-constant and distinct roots
\[\lmd^1(u), ...,\lmd^n(u).\]
These functions give a system of local coordinates of $M$, and we call them the canonical coordinates. In the canonical coordinates, the functions $g_1^{ij}$, $g_2^{ij}$ have the forms
\[g_1^{ij}(\lm)=\delta^{ij} f^i, \quad g_2^{ij}(\lm)=\lm_i\delta^{ij} f^i.\]
Let us still use $\mathcal{Q}^{ij}_{a;s,k}$ to denote the coefficients of $\mathcal{Q}_a^{ij}$ in canonical coordinates, then the central invariants of
$(\mcalQ_1,\mcalQ_2)$ are defined by
\begin{equation}\label{defn-of-cen-inv}
  c_i(\lmd)
 =\frac{1}{3(f^i)^2}
  \left(
    \mcalQ^{ii}_{2;2,3}
   -\lmd^i\mcalQ^{ii}_{1;2,3}
   +\sum_{k\neq i}
    \frac{(\mcalQ^{ki}_{2;1,2}-\lmd^i\mcalQ^{ki}_{1;1,2})^2}
         {f^k(\lmd^k-\lmd^i)}
  \right).
\end{equation}
In fact, the $i$-th central invariant $c_i(\lmd)$ depends only on $\lmd^i$,
see \cite{DLZ-2006}.


Now we consider the central invariants of the
bihamiltonian structure $(\mcalP_1, \mcalP_2)$ of the Albowitz-Ladik hierarchy
defined by \eqref{tpo-1-x} and \eqref{tpo-2-x}.
In the coordinates $u^1=P, u^2=Q$, the Hamiltonian operators can be represented, after the rescaling $\mathcal{P}_a\to \frac1{\ve} \mathcal{P}_a$, as follows:
\begin{align*}
\mcalP_1
&=
  \mcalP_{1;0}
 +\ve
  \begin{pmatrix}
    -\frac12 u^2_{xx}-u^2_x\p_x & -A^* \vsp\\
    A& 0
  \end{pmatrix}
 +\ve^2
  \begin{pmatrix}
    B
  & -C^*  \\
    C & 0
  \end{pmatrix}
 +O(\ve^3),
\\
\mcalP_2
&=
  \mcalP_{2;0}
 +\ve
  \begin{pmatrix}
    0 & -D^* \\
    D
    & 0
  \end{pmatrix}  +\ve^2
  \begin{pmatrix}
    0
  & -E^* \\
    E
  & F
  \end{pmatrix}
 +O(\ve^3),
\end{align*}
where
\begin{equation}\label{leadingtm}
  \mcalP_{1;0}
 =\begin{pmatrix}
    -\p_x u^2-u^2\p_x & -\p_x u^2 \vsp\\
    -u^2\p_x & 0
  \end{pmatrix},\quad
  \mcalP_{2;0}
 =\begin{pmatrix}
    0 & u^1\p_x u^2 \vsp\\
    u^2\p_x u^1 & 2u^2\p_x u^2
  \end{pmatrix},
\end{equation}
and the operators $A, B, C, D, E, F$ are given by
\begin{align*}
A&=\frac12u^2\p_x^2,\quad
B=-\frac16 u^2_{xxx}-\frac{u^2_{xx}}{2}\p_x-\frac12 u^2_x\p_x^2-\frac13 u^2\p_x^3,\\
C&=-\frac16 u^2\p_x^3,\quad
D= -\frac12 u^1_{xx}u^2 -u^1_xu^2\p_x-\frac12 u^1u^2\p_x^2,\\
E&= \frac16 u^1_{xxx}u^2+\frac12 u^1_{xx}u^2\p_x+\frac12 u^1_xu^2\p_x^2+\frac16 u^1u^2\p_x^3,\\
F&=\frac13 u^2u^2_{xxx}+u^2u^2_{xx}\p_x+u^2u^2_x\p_x^2+\frac13 (u^2)^2\p_x^3.
\end{align*}
The canonical coordinates $(\lmd^1,\lmd^2)$
of $(\mcalP_{0,1},\mcalP_{0,2})$ are given by
\begin{eqnarray*}
  \lmd^1&=&u^1-2u^2+2\sqrt{(u^2)^2-u^1u^2}, \\
  \lmd^2&=&u^1-2u^2-2\sqrt{(u^2)^2-u^1u^2},
\end{eqnarray*}
and the Jacobian between the coordinate systems $(u^1,u^2)$ and $(\lmd^1, \lmd^2)$ is given by
\[
  \mathcal J
 =\begin{pmatrix}
    \frac{\p\lmd^1}{\p u^1} & \frac{\p\lmd^1}{\p u^2} \vsp\\
    \frac{\p\lmd^2}{\p u^1} & \frac{\p\lmd^2}{\p u^2}
  \end{pmatrix}
 =
  \begin{pmatrix}
    -1-\frac{u^2}{\sqrt{(u^2)^2-u^1u^2}} & 2+\frac{-u^1+2u^2}{\sqrt{(u^2)^2-u^1u^2}} \vsp\\
    -1+\frac{u^2}{\sqrt{(u^2)^2-u^1u^2}} & 2-\frac{-u^1+2u^2}{\sqrt{(u^2)^2-u^1u^2}}
  \end{pmatrix}.
\]
Thus in canonical coordinate $(\lmd^1,\lmd^2)$ we have
\[
  \begin{pmatrix}
    f^1 & 0 \\
    0 & f^2
  \end{pmatrix}
 =\mcalJ\cdot \begin{pmatrix}-2 u^2&- u^2\\ -u^2&0\end{pmatrix}\cdot\mcalJ^T,\]
  and
\begin{align*}
\mcalP_{1;1,2}
&=
  \mcalJ
  \begin{pmatrix}
    0 & -\frac12 u^2 \vsp\\
    \frac12 u^2 & 0
  \end{pmatrix}\mcalJ^T,\quad
\mcalP_{2;1,2}
=
 \mcalJ
  \begin{pmatrix}
    0 & \frac12 u^1u^2 \vsp\\
   -\frac12 u^1u^2 & 0
  \end{pmatrix}\mcalJ^T\\ 
\mcalP_{1; 2,3}&=
\mcalJ
  \begin{pmatrix}
   - \frac13 u^2 & -\frac16 u^2 \vsp\\
    -\frac16 u^2 & 0
  \end{pmatrix}\mcalJ^T,\quad
\mcalP_{2;2,3}=
\mcalJ
  \begin{pmatrix}
    0 & \frac16 u^1u^2 \vsp\\
    \frac16 u^1u^2 & \frac13 (u^2)^2
  \end{pmatrix}\mcalJ^T.
\end{align*}
Then it follows from the formula \eqref{defn-of-cen-inv} that
the central invariants of bihamiltonian structure $(\mcalP_1,\mcalP_2)$ are given by
\[
  c_1=c_2=\frac{1}{24}.
\]

In the same way, we can compute the central invariants of the other bihamiltonian structures that come from the tri-Hamiltonian structure $(\mathcal{P}_1, \mathcal{P}_2, \mathcal{P}_3)$, and we obtain the following theorem.
\begin{thm}
The central invariants of the bihamiltonian structures $(\mathcal{P}_i, \mathcal{P}_j)$ of the Ablowitz-Ladik hierarchy are given by the following table:
\begin{table}[htbp]
\centering
\begin{tabular}{ccc}
\toprule
 &\quad\qquad  $c_1$ & \quad\qquad  $c_2$\\
\midrule
$(\mcalP_1, \mcalP_2)$ & \quad\qquad  $\frac1{24}$  &\quad\qquad  $\frac1{24}$ \\ [0pt]
\midrule
$(\mcalP_2, \mcalP_1)$ & \quad\qquad  $-\frac1{24\lmd^1}$ &\quad\qquad   $-\frac1{24\lmd^2}$\\ [0pt]
\midrule
$(\mcalP_1, \mcalP_3)$ &  \quad\qquad  $-\frac1{48\sqrt{\lmd^1}}$ & \quad\qquad  $-\frac1{48\sqrt{\lmd^2}}$\\ [0pt]
\midrule
$(\mcalP_3, \mcalP_1)$ &   \quad\qquad  $-\frac1{48\sqrt{\lmd^1}}$ &\quad\qquad  $-\frac1{48\sqrt{\lmd^2}}$\\ [0pt]
\midrule
$(\mcalP_2, \mcalP_3)$ &   \quad\qquad   $\frac1{24\lmd^1}$ &\quad\qquad  $\frac1{24\lmd^2}$\\ [0pt]
\midrule
$(\mcalP_3, \mcalP_2)$ &   \quad\qquad  $-\frac1{24}$ &\quad\qquad  $-\frac1{24}$\\
\bottomrule
\end{tabular}
\end{table}
\end{thm}

We will consider the relation between the bihamiltonian structures $(\mathcal{P}_1, \mathcal{P}_2)$ and $(\mathcal{P}_3, \mathcal{P}_2)$ in Sec.\,\ref{sec-7}.

\section{Relation between the positive and negative flows}
In this section, we establish a relation between the positive and negative flows of the Ablowitz-Ladik hierarchy. To this end, let us first note
that the Lax equations \eqref{al-def}, \eqref{al-def-n}
are the compatibility conditions of the linear system \eqref{sp-2}
with the following evolutions of the eigenfunction:
\begin{align}
\phi_{t_k}&=\frac1{(k+1)!} (L^{k+1})_+\phi=a_k(z) \phi+b_k(z)\phi^-,\label{evo-phi-1}\\
\phi_{s_k}&=\frac1{(k+1)!} (M^{k+1})_-\phi=c_k(z^{-1}) \phi^++d_k(z^{-1})\phi^-.\label{evo-phi-2}
\end{align}
Here $a_k(z), c_k(z)$ are polynomials of $z$ of order $k+1$, and $b_k(z), d_k(z)$ are polynomials of $z$ of order $k$. These polynomials are determined, up to some integral constants, by the compatibility conditions between \eqref{sp-2} and \eqref{evo-phi-1}, \eqref{evo-phi-2}.

We now perform a gauge transformation $\psi=\rho \phi$
to the spectral problem \eqref{sp-2}, where the function $\rho$ is
defined by the relation
\begin{equation*}
\rho P=\rho^- Q.
\end{equation*}
Then the spectral problem \eqref{sp-2} is transformed to
\begin{equation}
\hat{L}\psi=z^{-1}\psi,\label{sp-3}
\end{equation}
where
\[\hat{L}=(1-\hat{Q}{\hat{\Lmd}}^{-1})^{-1}(\hat{\Lmd}-\hat{P}),
\quad \hat{\Lmd}=\Lmd^{-1},\]
and
\beq\label{transform}
\hat{P}=\frac1{P},\quad \hat{Q}=\frac{Q^+}{P P^+}.
\eeq
In term of the new eigenfunction $\psi$, the linear systems \eqref{evo-phi-1} and \eqref{evo-phi-2} can be represented as
\begin{align}
\psi_{t_k}&=(\frac{\rho_{t_k}}\rho+a_k(z))\psi+\frac{\rho}{\rho^-}\psi^-
=\hat{c}_k(z) \psi+\hat{d}_k(z)\psi^-,\label{evo-psi-1}\\
\psi_{s_k}&=(\frac{\rho_{s_k}}\rho+a_k(z))\psi+\frac{\rho}{\rho^-}\psi^-
=\hat{a}_k(z^{-1}) \psi+\hat{b}_k(z^{-1})\psi^-.\label{evo-psi-2}
\end{align}
From the compatibility condition of these linear systems with \eqref{sp-3} it follows that
\[\hat{a}_k(z)=a_k(z),\quad \hat{b}_k(z)=b_k(z),\quad \hat{c}_k(z)=c_k(z),\quad \hat{d}_k(z)=d_k(z).\]
Thus, if we denote the Ablowitz-Ladik hierarchy \eqref{al-def}, \eqref{al-def-n} in terms of the unknown functions $P, Q$ as follows
\begin{align*}
\frac{\p P}{\p t_k}&=T_{1,k}(P, Q, P^-, Q^-, P^+, Q^+,\dots),\\
\frac{\p Q}{\p t_k}&=T_{2,k}(P, Q, P^-, Q^-, P^+, Q^+,\dots),\\
\frac{\p P}{\p s_k}&=S_{1,k}(P, Q, P^-, Q^-, P^+, Q^+,\dots),\\
\frac{\p Q}{\p s_k}&=S_{2,k}(P, Q, P^-, Q^-, P^+, Q^+,\dots),
\end{align*}
then from the compatibility condition of \eqref{sp-3} with \eqref{evo-psi-1}, \eqref{evo-psi-2} that the Ablowitz-Ladik hierarchy \eqref{al-def}, \eqref{al-def-n} can also be represented in terms of the unknown functions $\hat{P}, \hat{Q}$ as follows:
\begin{align*}
\frac{\p \hat{P}}{\p t_k}&=S_{1,k}(\hat{P}, \hat{Q}, \hat{P}^+, \hat{Q}^+, \hat{P}^-, \hat{Q}^-,\dots),\\
\frac{\p \hat{Q}}{\p t_k}&=S_{2,k}(\hat{P}, \hat{Q}, \hat{P}^+, \hat{Q}^+, \hat{P}^-, \hat{Q}^-,\dots),\\
\frac{\p \hat{P}}{\p s_k}&=T_{1,k}(\hat{P}, \hat{Q}, \hat{P}^+, \hat{Q}^+, \hat{P}^-, \hat{Q}^-,\dots),\\
\frac{\p \hat{Q}}{\p s_k}&=T_{2,k}(\hat{P}, \hat{Q}, \hat{P}^+, \hat{Q}^+, \hat{P}^-, \hat{Q}^-,\dots).
\end{align*}
Thus, after the change of the unkown functions \eqref{transform} and the shift operator $\clm\to \clm^{-1}$, the positive and negative flows of the Ablowitz-Ladik hierarchies are interchanged. For example, the $\frac{\p}{\p t_0}$-flow of the Ablowitz-Ladik hierarchy \eqref{exm-1} can be written as
\[\hat{P}_{t_0}=\frac{\hat{Q}^-}{\hat{P}^-}-\frac{\hat{Q}}{\hat{P}^+},
\quad \hat{Q}_{t_0}=\frac{\hat{Q}}{\hat{P}}-\frac{\hat{Q}}{\hat{P}^+},
\]
which has the same form as that of the $\frac{\p}{\p s_0}$-flow given in \eqref{exm-5}.

We note that the original Ablowitz-Ladik equation corresponds to the following system of equations for the unknown functions $P, Q$:
\begin{align*}
\frac{\p P}{\p t}&=\frac{\p P}{\p t_0}+\frac{\p P}{\p s_0}=P(Q^+-Q)+\frac{Q^+}{P^+}-\frac{Q}{P^-},\\
\frac{\p Q}{\p t}&=\frac{\p Q}{\p t_0}+\frac{\p Q}{\p s_0}=Q (Q^+-Q^--P+P^-)+\frac{Q}{P}-\frac{Q}{P^-}.
\end{align*}
From the above-mentioned relation between the $\frac{\p}{\p t_0}$-flow and the $\frac{\p}{\p s_0}$-flow it follows that the original Ablowitz-Ladik equation has the following B\"acklund transformation:
\[P(n,t)\to \frac1{P(-n, t)},\quad Q(n, t)\to \frac{Q(-n+1,t)}{P(-n,t)P(-n+1,t)}.\]

\section{Relation between the bihamiltonian structures}\label{sec-7}
To see the relation between the bihamiltonian structures $({\mathcal{P}}_1, {\mathcal{P}}_2)$ and $({\mathcal{P}}_3, {\mathcal{P}}_2)$, let us represent them in terms of the unknown functions $\hat{P}, \hat{Q}$ defined by \eqref{transform}.

Denote
\[\mathcal{J}=\begin{pmatrix}
-\frac1{P^2} &0\vsp\\ -\frac{Q^+}{P^2 P^+}-\frac{Q^+}{P (P^+)^2}\Lmd& \frac1{P P^+}\Lmd
\end{pmatrix},\quad \mathcal{J}^*=\begin{pmatrix}
-\frac1{P^2} &-\frac{Q^+}{P^2 P^+}-\frac{Q}{P^2 P^-}\Lmd^{-1}\vsp\\ 0& \frac1{P P^-}\Lmd^{-1}
\end{pmatrix},\]
then by a straightforward computation we known that, in terms of $\hat{P}, \hat{Q}$, the Hamiltonian operator ${\mathcal{P}}_2$ has the expression
\begin{align*}
\hat{\mathcal{P}}_2&=\mathcal{J}\mathcal{P}_2 \mathcal{J}^*=
\begin{pmatrix}0 & \hat{P}(\clm^{-1}-1) \hat{Q}\vsp\\ \hat{Q}(1-\clm) \hat{P} & \hat{Q}(\clm^{-1}-\clm) \hat{Q}\end{pmatrix}\\
&=\left.\mathcal{P}_2\right|_{P\to \hat{P},\, Q\to\hat{Q}, \,\Lmd^\pm\to \Lmd^\mp},
\end{align*}
and the Hamiltonian operators ${\mathcal{P}}_1$,  ${\mathcal{P}}_3$ have the expressions
\begin{align*}
\hat{\mathcal{P}}_1&=\mathcal{J}\mathcal{P}_1 \mathcal{J}^*=-\left.\mathcal{P}_3\right|_{P\to \hat{P},\, Q\to\hat{Q}, \,\Lmd^\pm\to \Lmd^\mp},\\
\hat{\mathcal{P}}_3&=\mathcal{J}\mathcal{P}_3 \mathcal{J}^*=-\left.\mathcal{P}_1\right|_{P\to \hat{P},\, Q\to\hat{Q}, \,\Lmd^\pm\to \Lmd^\mp}.
\end{align*}
Thus the bihamiltonian structure $({\mathcal{P}}_1, {\mathcal{P}}_2)$ is equivalent to $(-{\mathcal{P}}_3, {\mathcal{P}}_2)$ under the transformation \eqref{transform} and the transformation $\ve\to -\ve$.

\section{Conclusion}
We give in this paper a tri-Hamiltonian structure of the Ablowitz-Ladik hierarchy, which yields in particular a local bihamiltonian structure with the dispersionless limit given in \cite{Brini-2}. This dispersionless limit is a bihamiltonian structure of hydrodynamic type defined on the jet space of the generalized 2-dimensional Frobenius manifold given by the potential
\[F=\frac12 (v^1)^2 v^2+v^1 e^{v^2}+\frac12 (v^1)^2\log v^1\]
and the flat metric
\[\eta=(\eta_{\al\beta})=\begin{pmatrix} 0 & 1\\ 1&0\end{pmatrix}.\]

The above Frobenius manifold has a non-constant unity vector field
\[e=\frac{v^1}{v^1-e^{v^2}}\frac{\p}{\p v^1}-\frac{1}{v^1-e^{v^2}}\frac{\p}{\p v^2},\]
and the Euler vector field
\[E=v^1\frac{\p}{\p v^1}+\frac{\p}{\p v^2}.\]
Its intersection form is given by
\[g=(g^{\al\beta})=\left(E^\gamma\eta^{\al\xi}\eta^{\beta\zeta}\frac{\p^3 F}{\p v^\gamma\p v^\xi\p v^\zeta}\right)=\begin{pmatrix}2 v^1 e^{v^2} & v^1+ e^{v^2}\vsp\\ v^1+ e^{v^2} & 2\end{pmatrix}.\]
The flat metric $\eta$ and the intersection form $g$ induce a bihamiltonian structure of hydrodynamic type $(\tilde{\mathcal{P}}_1, \tilde{\mathcal{P}}_2)$
with the Hamiltonian operators defined by
\[\tilde{\mathcal{P}}_1=\begin{pmatrix}0& \p_x\\ \p_x&0\end{pmatrix},\quad
\tilde{\mathcal{P}}_2=\begin{pmatrix}2 v^1 e^{v^2}\p_x+(v^1 e^{v^2})'& (v^1+e^{v^2})\p_x\vsp\\
(v^1+e^{v^2})\p_x+(v^1+e^{v^2})'& 2\p_x\end{pmatrix}.
\]
It coincides with the dispersionless limit $(\mathcal{P}_{1;0}, \mathcal{P}_{2;0})$ of the bihamiltonian structure $(\mathcal{P}_1, \mathcal{P}_2)$ of the Ablowitz-Ladik hierarchy given in \eqref{leadingtm} under the following change of the unkown functions \eqref{vpq}, i.e.
\[v^1=u^2-u^1,\quad v^2=\log u^2.\]
Under this change of the unkown functions, the dispersionless limits of the positive flows $\frac{\p}{\p t_k}$ $(k\ge 0)$ of the Ablowitz-Ladik hiearrchy also coincide with the following flows $\frac{\p}{\p t^{2,k}}$ of the Principal Hierarchy of the Frobenius manifold \cite{Brini-2}:
\[\frac{\p v^\al}{\p t^{2,k}}=\eta^{\al\beta}\frac{\p^2\theta_{2,k+1}}{\p v^\beta\p v^\gamma} v^\gamma_x,\quad k\ge 0,\]
where
\begin{align*}
&\theta_{2,0}=v^1,\quad \theta_{2,1}=v^1 e^{v^2}+\frac12 (v^1)^2,\quad \theta_{2,2}=\frac12 v^1 e^{2 v^2}+(v^1)^2 e^{v^2}+\frac16 (v^1)^3,
\end{align*}
and in general, we have
\begin{align*}
\theta_{2,k}&=\frac{(-1)^{k+1} (e^{v^2}-v^1)^{k+1}}{(k+1)!}\, {}_2F_1\left(-k-1,k+1;1;\frac{e^{v^2}}{e^{v^2}-v^1}\right)\\
&=\frac1{(k+1)!}\sum_{s=0}^{k+1}\binom{k+1}{s}\binom{k+s}{s} e^{s v^2} (v^1-e^{v^2})^{k+1-s}.
\end{align*}
Define the recursion operator
\[\tilde{\mathcal{R}}=\tilde{\mathcal{P}}_2\circ \tilde{\mathcal{P}}_1^{-1}
=\begin{pmatrix}v^1+e^{v^2}&2 v^1 e^{v^2}+(v^1 e^{v^2})'\p_x^{-1}\vsp\\
2& v^1+e^{v^2}+(v^1+e^{v^2})'\p_x^{-1}\end{pmatrix},\]
then these flows satisfy the recursion relation
\[\begin{pmatrix}\frac{\p v^1}{\p t^{2,k}} \vvsp\\ \frac{\p v^2}{\p t^{2,k}}\end{pmatrix}=\frac1{k+1}\tilde{\mathcal{R}} \begin{pmatrix}\frac{\p v^1}{\p t^{2,k-1}}\vvsp \\ \frac{\p v^2}{\p t^{2,k-1}}\end{pmatrix},\quad k\ge 1.\]

The dispersionless limit of the negative flows $\frac{\p}{\p s_k}$ $(s\ge 0)$
can be represented in the form
\[\frac{\p v^\al}{\p s_k}=\eta^{\al\beta}\frac{\p^2 h_{k}}{\p v^\beta\p v^\gamma} v^\gamma_x,\quad k\ge 0,\]
where
\[h_0=v^2-\log(v^1-e^{v^2}),\quad
h_k=\frac1{k(k+1)}\frac{\theta_{2,k-1}}{(v^1-e^{v^2})^{2k}}.
\]
These flows satisfy the recursion relation
\[\tilde{\mathcal{R}} \begin{pmatrix}\frac{\p v^1}{\p s_0} \vvsp\\ \frac{\p v^2}{\p s_0}\end{pmatrix}=0,\quad \begin{pmatrix}\frac{\p v^1}{\p s_{k-1}} \vvsp\\ \frac{\p v^2}{\p s_{k-1}}\end{pmatrix}=(k+1)\tilde{\mathcal{R}} \begin{pmatrix}\frac{\p v^1}{\p s_k} \vvsp\\ \frac{\p v^2}{\p s_k}\end{pmatrix},\quad k\ge 1.\]

We note that the flows $\frac{\p}{\p s_k}$ $(s\ge 0)$ do not belong to the Principal Hierarchy of the Frobenius manifold, which consists of the flows $\frac{\p}{\p t^{2,k}}$ and $\frac{\p}{\p t^{1,k}}$ $(k\ge 0)$. As it was shown in \cite{Brini-2},
the flows $\frac{\p}{\p t^{1,k}}$ are given by the equations
\[\frac{\p v^1}{\p t^{1,0}}=v^1_x+e^{v^2} v^2_x,\quad
\frac{\p v^2}{\p t^{1,0}}=\frac{v^1_x}{v^1}+v^2_x\]
and the following recursion relation:
\[\begin{pmatrix}\frac{\p v^1}{\p t^{1,k}} \vsp\\ \frac{\p v^2}{\p t^{1,k}}\end{pmatrix}=\frac1{k}\tilde{\mathcal{R}} \begin{pmatrix}\frac{\p v^1}{\p t^{1,k-1}} \vvsp\\ \frac{\p v^2}{\p t^{1,k-1}}\end{pmatrix}-\frac2{k}\begin{pmatrix}\frac{\p v^1}{\p t^{2,k-1}} \vvsp\\ \frac{\p v^2}{\p t^{2,k-1}}\end{pmatrix},\quad k\ge 1.\]
We will consider in subsequent publications a certain extension of the Ablowitz-Ladik hierarchy such that its dispersionless limit contains the whole Principal Hierarchy of the above-mentioned Frobenius manifold, and study Brini's conjecture on its relation with the Gromov-Witten invariants of local $\mathbb{CP}^1$.
\vskip 0.3truecm

\noindent{\bf Acknowledgements.}
This work is supported by NSFC No.\,11771238 and No.\,11725104.

\end{document}